\documentclass[runningheads,a4paper]{llncs}

\usepackage{url}
\newcommand{\keywords}[1]{\par\addvspace\baselineskip
\noindent\keywordname\enspace\ignorespaces#1}

\usepackage{amsfonts}
\usepackage{amssymb}
\usepackage{graphicx}
\usepackage{enumerate}
\usepackage{amsmath}

\newcommand{\beeq}[1]{\begin{equation} \label{#1}}
\newcommand{\eeq}{\end{equation}}
\newcommand{\beeqs}{\begin{eqnarray*}}
\newcommand{\eeqs}{\end{eqnarray*}}
\renewcommand{\(}{\begin{eqnarray*}}
\renewcommand{\)}{\end{eqnarray*}}
\newcommand{\beeqn}{\begin{eqnarray}}
\newcommand{\eeqn}{\end{eqnarray}}

\newcommand{\nexteqline}{\\ &=&}

\newcommand{\refth}[1]{Theorem~\ref{#1}}
\newcommand{\reflm}[1]{Lemma~\ref{#1}}

\newcommand{\refprop}[1]{Proposition~\ref{#1}}

\newcommand{\refeq}[1]{(\ref{#1})}

\renewcommand{\quad}{\hspace*{3mm}}
\renewcommand{\qquad}{\hspace*{5mm}}

\newcommand{\lp}{\left(  }
\newcommand{\rp}{\right) }
\newcommand{\lb}{\left\{  }
\newcommand{\rb}{\right\} }

\newfont{\Bb}{msbm8 scaled\magstep1}
\newcommand{\E}{\mbox{\Bb E}}

\newcommand{\R}{\mbox{\Bb R}}

\newcommand{\ep}{\epsilon}

\newcounter{cnt1}
\newcounter{cnt2}
\newcounter{cnt3}
\newcounter{cnt4}
\newcounter{cnt5}
\newcounter{cnta}

\newcommand{\beenu}
{
\begin{list}{\arabic{cnt1}.}
{\usecounter{cnt1}
\leftmargin 4mm
\setlength{\leftmargin}{\leftmargin}
\topsep 0pt
\parsep 0pt
\itemsep 0pt}
}
\newcommand{\eenu}{\end{list}}

\newcommand{\beenub}
{
\begin{list}{\arabic{cnt1}-\arabic{cnt2}.}
{
\leftmargin  4mm
\setlength{\leftmargin}{\leftmargin}
\topsep 0pt
\parsep 0pt
\itemsep 0pt}
}
\newcommand{\eenub}{\end{list}}

\newcommand{\beenuc}
{
\begin{list}{\arabic{cnt1}-\arabic{cnt2}-\arabic{cnt3}.}
{
\leftmargin  4mm
\setlength{\leftmargin}{\leftmargin}
\topsep 0pt
\parsep 0pt
\itemsep 0pt}
}
\newcommand{\eenuc}{\end{list}}

\newcommand{\beenud}
{
\begin{list}{\arabic{cnt1}-\arabic{cnt2}-\arabic{cnt3}-\arabic{cnt4}.}
{
\leftmargin  4mm
\setlength{\leftmargin}{\leftmargin}
\topsep 0pt
\parsep 0pt
\itemsep 0pt}
}
\newcommand{\eenud}{\end{list}}

\newcommand{\beenue}
{
\begin{list}{\arabic{cnt1}-\arabic{cnt2}-\arabic{cnt3}-\arabic{cnt4}-\arabic{cnt5}.}
{
\leftmargin  4mm
\setlength{\leftmargin}{\leftmargin}
\topsep 0pt
\parsep 0pt
\itemsep 0pt}
}
\newcommand{\eenue}{\end{list}}

\newcommand{\beitm}
{
\begin{list}{$\bullet$}
{
\leftmargin  2mm
\setlength{\leftmargin}{\leftmargin}
\topsep 3pt
\parsep 0pt
\itemsep 2pt}
}
\newcommand{\eitm}{\end{list}}

\newcommand{\beitma}
{
\begin{list}{(\alph{cnta})~~}
{\usecounter{cnta}
\labelsep 0mm
\leftmargin  2mm
\setlength{\leftmargin}{\leftmargin}
\topsep 0pt
\parsep 0pt
\itemsep 0pt}
}

\newcommand{\broman}
{
\begin{list}{\roman{cnt1})}
{
\usecounter{cnt1}
\leftmargin 5mm
\setlength{\leftmargin}{\leftmargin}
\topsep 0pt
\parsep 0pt
\itemsep 0pt}
}
\newcommand{\eroman}{\end{list}}

\newcommand{\bRoman}
{
\begin{list}{\Roman{cnt1})}
{
\usecounter{cnt1}
\leftmargin 5mm
\setlength{\leftmargin}{\leftmargin}
\topsep 0pt
\parsep 0pt
\itemsep 0pt}
}
\newcommand{\eRoman}{\end{list}}

\newcommand{\balph}
{
\begin{list}{\alph{cnt1})}
{
\usecounter{cnt1}
\leftmargin 5mm
\setlength{\leftmargin}{\leftmargin}
\topsep 0pt
\parsep 0pt
\itemsep 0pt}
}
\newcommand{\ealph}{\end{list}}

\newcommand{\bAlph}
{
\begin{list}{\Alph{cnt1})}
{
\usecounter{cnt1}
\leftmargin 3mm
\setlength{\leftmargin}{\leftmargin}
\topsep 0pt
\parsep 0pt
\itemsep 0pt}
}
\newcommand{\eAlph}{\end{list}}

\newcommand{\bbullet}
{
\begin{list}{$\bullet$}
{
\leftmargin  4mm
\setlength{\leftmargin}{\leftmargin}
\topsep 3pt
\parsep 0pt
\itemsep 2pt}
}
\newcommand{\ebullet}{\end{list}}

\newcommand{\bdash}
{
\begin{list}{-}
{
\leftmargin 4mm
\setlength{\leftmargin}{\leftmargin}
\topsep 0pt
\parsep 0pt
\itemsep 0pt}
}
\newcommand{\edash}{\end{list}}

\begin{document}

\mainmatter  % start of an individual contribution

\title{Partial-Match Queries with Random Wildcards: In Tries and Distributed Hash Tables}

\titlerunning{Partial-Match Queries with Random Wildcards}

\author{Junichiro Fukuyama}
\authorrunning{Junichiro Fukuyama}

\institute{Applied Research Laboratory, The Pennsylvania State University\\
jxf140@psu.edu}


\toctitle{Partial-Match Queries with Wildcards}
\tocauthor{}
\maketitle

\begin{abstract}
Consider an $m$-bit query $q$ to a bitwise trie $T$. A {\em wildcard $*$} is an unspecified bit in $q$ for which the query asks the membership for both cases $*=0$ and $*=1$. It is common that such {\em partial-match queries} with wildcards are issued in tries. 
With uniformly random occurrences of $w$ wildcards in $q$ assumed, the obvious upper bound on the average number of traversal steps in $T$ is $2^w m$. We show that the average does not exceed 
\[
\frac{m+1}{w+1} \lp 2^{w+2} - 2 w - 4 \rp  + m  = 
O \lp \frac{2^w  m}{w} \rp,
\]
and equals the value exactly when $T$ includes all the $m$-bit keys as the worst case. Here the query $q$ performs with the naive backtracking algorithm in $T$. It is similarly shown that the average is $O \lp \frac{k^w  m}{w} \rp$ in a general trie of maximum out-degree $k$. 
Our analysis for tries is extended to a distributed hash table (DHT), which is among the most frequently used decentralized data structures in networking. 
We show, under a natural probabilistic assumption for the largest class of DHTs, that the average number of hops required by an $m$-bit query $q$ to a DHT $D$ with random $w$ wildcards meets the same asymptotic bound. As a result, $q$ is answered with average $O \lp \frac{2^w  m}{w} \rp$ hops rather than $\Theta \lp 2^w m \rp$ in the four major DHTs Chord, Pastry, Tapestry and Kademlia. In addition, with a uniform key distribution for sufficiently many entries, we prove that a lookup request to the DHT Chord is answered correctly with $O \lp m \rp$ hops and probability $1 - 2^{-\Omega \lp m \rp}$. To the author's knowledge, the probability $1 - 2^{-\Omega \lp m \rp}$ of correct lookup in Chord has not been identified so far.

\end{abstract}

\section{Introduction} \label{Introduction}

Finding information that partially matches to a given pattern has been a major problem in computer science for decades. In addition to the classical RK and KMP-algorithms in textbooks such as \cite{Textbook}, a collection of research results on {\em partial-match queries} is found in literature such as \cite{R76,CIP02}. It is common in practice to construct a trie 
as the data structure for partial-match queries with wildcards \cite{Sedgewick}. 
Here a {\em trie} is the well-known prefix tree data structure to store keys \cite{Knuth}, used for applications including dictionary search and lexicographic sorting. The most basic form of a trie $T$ is the {\em bitwise trie} to store $m$-bit integer keys. Denote by $q$ a query to such $T$. 
A {\em wildcard $*$} in $q$ is defined as an unspecified bit for which $q$ asks the membership for both cases $*=0$ and $*=1$. For example, $q=1*0*0$ is a 5-bit query asking if $10000$, $10010$, $11000$ and $11010$ are in $T$.

In the paper, we first analyze the average performance of an $m$-bit query $q$ to $T$ with random $w$ wildcards. Assume that $T$ is a bitwise trie for which the wildcards occur in $w$ positions in $q$, chosen randomly with the uniform probability density function (PDF), and also that $q$ performs with the naive backtracking algorithm. 
We show that the average number of steps in $T$ required by $q$ does not exceed 
\[
\frac{m+1}{w+1} \lp 2^{w+2} - 2 w - 4 \rp  + m  = 
O \lp \frac{2^w  m}{w} \rp,
\]
and is exactly equal to the value when $T$ includes all the $m$-bit keys as the worst case. 
This improves the obvious upper bound $2^w m$ asymptotically.  
We will also prove that the average is $O \lp \frac{k^w  m}{w} \rp$ in a general trie of maximum out-degree $k$.  The results have been unknown so far despite the common use of queries with wildcards in $T$. In Section 4, we will present an example of a practical system in which the above analysis could be useful.

The second half of the paper extends our analysis to a distributed hash table (DHT), which is among the most significant decentralized data structures used in networking. A DHT  can support a number of application services such as web caching, file sharing, name-address mapping to track node mobility \cite{DMap}, instant messaging, multicast, content distribution, etc. 
In \cite{DHTAnalysis}, detailed analysis is presented on the tradeoff between the routing table size and average number of hops per lookup ({\em network diameter}) in different DHTs. In the taxonomy, the class of DHTs with $O \lp \log n \rp$ routing table size and $O \lp \log n \rp$ network diameter is the largest one ($n$: the number of nodes in the DHT). We focus on this DHT class denoted by ${\cal C}$, which includes the four major DHTs Chord \cite{Chord}, Pastry \cite{Pastry}, Tapestry \cite{Tapestry}, and Kademlia \cite{Kademlia}.

We will see a structural similarity between a bitwise trie and DHT in order to answer an $m$-bit query. 
With the above $O \lp \frac{2^w  m}{w} \rp$ bound for bitwise tries and another probabilistic assumption, we show that the average number of hops required by an $m$-bit query $q$ to a DHT $D$ with random $w$ wildcards meets the same asymptotic bound. 
Arguing that the probabilistic assumption holds generally for the DHT class ${\cal C}$, we will especially confirm it for the above four DHTs.  
The result thus improves the theoretical upper bound on the lookup time with $w$ random wildcards from $O \lp 2^w \log n \rp$ to $O \lp \frac{2^w  \log n}{w} \rp$ in the four DHTs.

In addition, with a uniform key distribution for $\Omega \lp mn \rp$ entries, we prove that a lookup request to the DHT Chord is answered correctly with $O \lp m \rp$ hops and  probability $1 - 2^{-\Omega \lp m \rp}$. The probability $1 - 2^{-\Omega \lp m \rp}$ of correct lookup in Chord will be identified for the first time to the author's knowledge.

The rest of the paper is structured as follows. In Section 2, we will prove the $O \lp \frac{2^w  m}{w} \rp$ and $O \lp \frac{k^w  m}{w} \rp$ bounds
for tries $T$. Section 3 shows  the $O \lp \frac{2^w  m}{w} \rp$ bound for the four DHTs, and the probability $1 - 2^{-\Omega \lp m \rp}$ of correct lookup in Chord. It is followed by concluding remarks in Section 4.

\section{Average Search Time with Random Wildcards in Tries} \label{Trie}

\subsection{In a Bitwise Trie}

An {\em $m$-bit query $q$ to a bitwise trie $T$ with $w$ wildcards} is a string consisting of $m-w$ 0s and/or 1s, and $w$ wildcards *. 
We assume that the letters in $q$ are numbered $m$, $m-1$, $\ldots$, $1$ from the left to right ({\em bit positions}).

We measure the running time of a query by the number of edges in $T$ traversed by the search algorithm, calling them {\em steps}. A query $q$ with $w$ wildcards completes in no more than $2^w m$ steps. We use the  standard $O$, $\Omega$ and $\Theta$-notations to express asymptotic quantities. A {\em constant} in this paper means a fixed positive real number depending on no other variable.

In this section, we prove that $q$ with $w$ random wildcards takes average $O \lp \frac{2^w m}{w} \rp$ steps. By $w$ random wildcards, we mean the following uniform assumption.

\medskip

\noindent
{\bf Assumption I: In an {\boldmath $m$}-bit query {\boldmath $q$} with {\boldmath $w$} wildcards,  $*$ occurs in {\boldmath $w$} positions with the uniform PDF.} 

\medskip

\noindent
In other words, every wildcard pattern, or {\em configuration}, occurs with the same probability $1 \bigr/ {m \choose w}$. Here a configuration determines the wildcard positions of a query $q$ to $T$. For example, c*cc* is a configuration in which $c$ represents 0/1. If $q$ satisfies Assumption I, it is said to be a {\em query $q$ to $T$ with uniformly random $w$ wildcards}.

Also consider the following natural backtracking search algorithm in $T$:

\medskip

\noindent
{\bf Algorithm {\sc Query}:} {\em Started at the root of $T$, search for the key such that every wildcard $*$ is 0 in $q$. When the current key's membership is determined, backtrack to the node of $T$ representing the closest unfinished wildcard\footnote{This means $*$ such that the search for $*=0$ is finished but $*=1$ is not.} * in $q$. Change $*=0$ into $*=1$. Search for the new key in $T$. Continue until the memberships of all the $2^w$ keys are determined.} 

\medskip

The intuition behind the proof of $O \lp \frac{2^w  m}{w} \rp$ steps is the following. If the wildcards in $q$ occur in bit positions bounded by a small integer $j \ge 1$, it takes at most $m+ O \lp 2^w j \rp$ steps to answer $q$, which is much smaller than $2^w m$. Since wildcards are placed randomly with the uniform PDF, this must affect the asymptotic number of steps required by $q$.

We prove the proposition below. It will be extended to a general trie of maximum out-degree $k$ in the next subsection. 

\begin{proposition} \label{Main}
Let $q$ be an $m$-bit query to a trie $T$ with uniformly random $w$ wildcards. Algorithm {\sc Query} answers $q$ in no more than 
\[
b= \frac{m+1}{w+1} \lp 2^{w+2} - 2 w - 4 \rp  + m 
\]
average steps, and in exactly $b$ average steps when $T$ includes all the possible $m$-bit keys. 
\qed
\end{proposition}

We construct its proof in what follows. We first show that the average is at most 
\beeq{smw}
s(m, w) \stackrel{def}{=} m + \sum_{\scriptstyle 1 \le z \le m \atop 1 \le j \le w} 
j 2^{w-j+1} \frac{{z \choose j}{m-z \choose w-j}}{{m \choose w}}. 
\eeq
The fraction ${z \choose j} {m-z \choose w-j}\Big/{m \choose w}$ is well-known as the {\em hypergeometric distribution} \cite{hyper}. It is the probability of $j$ successes in $w$ draws without replacement, from $m$ items including $z$ successes and $m-z$ failures.

Let $z_1$ denote the position of the rightmost wildcard bit, $i.e.$, the {\em least significant wildcard} in the given query $q$. Likewise, let $z_j$ be the position of the $j^{th}$ least significant wildcard. The set
\[
Z= \lb z_1, z_2, \ldots,  z_w \rb
\]
determines a configuration of $q$. Observe a lemma first for $q$ with a fixed configuration.

\begin{lemma} \label{Config1}
Algorithm {\sc Query} on a query $q$ with configuration $\lb z_1, z_2, \ldots,  z_w \rb$ terminates in 
\beeq{Estimate}
\hat{s} \lp m , w \rp \stackrel{def}{=} m + \sum_{j=1}^w 2^{w-j+1} z_j
\eeq
steps or less.
\end{lemma}
\begin{proof}
Prove the claim by induction on $w$. The basis occurs when $w=1$. One can check that $q$ with one wildcard in position $z_1$ takes at most $m+ 2 z_1= \hat{s} \lp m, 1 \rp$ steps, verifying the basis.

Assume true for $w-1$ and prove true for $w$. Below the stated number of steps are all in the worst case. The algorithm {\sc Query} first sets the most significant wildcard at $z_w$ as $*=0$, and performs the $m$-bit query with $w-1$ wildcards. It takes
$
\hat{s} \lp m, w-1 \rp
$
steps by induction hypothesis. Then it backtracks to the node representing the bit position $z_w$ with $z_w$ steps, set $*=1$, and recursively search for the remaining $w-1$ wildcards again. This takes extra $\hat{s}\lp z_w, w-1 \rp$ steps. 

So the total number of steps required by $q$ is at most 
\(
&&
\hat{s} \lp m, w - 1 \rp + z_w +  \hat{s} \lp z_w, w-1 \rp
\nexteqline
\lp m + \sum_{j=1}^{w-1} 2^{w-j} z_j \rp + z_w +
\lp z_w + \sum_{j=1}^{w-1} 2^{w-j} z_j \rp
\nexteqline
m + \sum_{j=1}^w 2^{w-j+1} z_j
=\hat{s} \lp m, w \rp,
\)
proving the induction step. The lemma follows.
\qed\end{proof}

Next we calculate the average of $\hat{s} \lp m, w \rp$ with the uniform occurrence of $Z=\lb z_1, z_2, \ldots,  z_w \rb$. 
Let $z\le m$ and $j \le w$ be positive integers. Denote by $p_{z,j}$ the probability that $z$ is the position of the $j^{th}$ least significant wildcard. We have the following lemma.

\begin{lemma}
$
p_{z, j} = 
{z-1 \choose j-1} {m-z \choose w-j}\Big/{m \choose w}
$.
\end{lemma}
\begin{proof}
Fix $z$ and $j$. The number of $Z$ such that $z_j=z$ is ${z-1 \choose j-1} {m-z \choose w-j}$. Since each configuration occurs with probability $1 \big/ {m \choose w}$, 
the probability of $z_j=z$ is $p_{z, j}$ as claimed.
\qed\end{proof}

If a given integer $z$ is $z_j$ in (\ref{Estimate}), it causes $2^{w-j+1} z_j =2^{w-j+1} z$ steps in the summation, which occurs with the probability $p_{z, j}$. The average number of steps required by $q$ is thus bounded by 
\(
&&
m + \sum_{\scriptstyle 1 \le z \le m \atop 1 \le j \le w} 2^{w-j+1} z p_{z, j}
=
m + \sum_{\scriptstyle 1 \le z \le m \atop 1 \le j \le w} 2^{w-j+1} z 
{z-1 \choose j-1} {m-z \choose w-j} \Big/ {m \choose w}
\nexteqline
m + \sum_{\scriptstyle 1 \le z \le m \atop 1 \le j \le w} j 2^{w-j+1}  
{z \choose j} {m-z \choose w-j} \Big/ {m \choose w}
=
s(m,w).
\)
This proves our claim that the algorithm 
{\sc Query} takes at most $s(m, w)$ steps on average. 
The bound is tight; when $T$ includes all the possible $m$-bit keys, {\sc Query} actually takes average $s(m,w)$ steps.

\medskip

It now suffices to show 
\beeq{eqMain}
s(m, w) = \frac{m+1}{w+1} \lp 2^{w+2} - 2 w - 4 \rp  + m, 
\eeq
to prove \refprop{Main}. As in standard textbooks on generating functions such as \cite{ConcreteMath}, for a function $f: (0,1) \rightarrow \R$ with its Taylor series, denote by $[x^k] f(x)$ the coefficient of $x^k$ in the series. Since ${z \choose j} = [x^z] \frac{x^j}{(1-x)^{j+1}}$ and ${m-z \choose w-j} = [x^{m-z}] \frac{x^{w-j}}{(1-x)^{w-j+1}}$, 
\(
&&
\sum_{1 \le z \le m} {z \choose j} {m-z \choose m-j} = 
\sum_{0 \le z \le m} {z \choose j} {m-z \choose m-j} 
\nexteqline
[x^m]  \frac{x^j}{(1-x)^{j+1}} \cdot \frac{x^{w-j}}{(1-x)^{w-j+1}}
=
[x^m]  \frac{x^w}{(1-x)^{w+2}} 
\nexteqline
[x^{m+1}]  \frac{x^{w+1}}{(1-x)^{w+2}} 
={m+1 \choose w+1}. 
\)

So, 
\(
s(m, w) &=& m + \sum_{\scriptstyle 1 \le j \le w} 
j 2^{w-j+1} \frac{{m+1 \choose w+1}}{{m \choose w}}
= \frac{m+1}{w+1} \lp 2^{w+2} - 2 w - 4 \rp  + m, 
\)
verifying \refeq{eqMain}. As we have already confirmed that the bound is tight, this completes the proof of \refprop{Main}.

\subsection{In a Trie of Maximum Out-Degree $k$}

%\pdffigure{fig1}{fig.png}{An Example of Backtracking in a Radix Tree}

We now consider a general trie $T$ of maximum out-degree $k$. We generalize \refprop{Main} into:

\begin{theorem} \label{RadixTree}
Let $T$ be a trie with maximum out-degree $k \ge 2$. 
A query to $T$ of length $m$ with $w$ uniformly random wildcards can be answered in average 
\[
b = m+ \frac{2(m+1)}{w+1} \cdot \frac{k^{w+1} - (w+1) k + w}{k-1}
\]
steps or less. The average is exactly $b$ when $q$ performs with Algorithm {\sc Query}, and 
$T$ is a complete $k$-ary tree. 
\qed
\end{theorem}
This means $q$ requires $O \lp \frac{k^w m}{w} \rp$ steps in $T$ as claimed in the introduction.

A general trie $T$ is formally defined with its {\em membership}: It is a tree such that each edge is associated with a letter in a given set $A$ ({\em alphabet}). A string $s \in A^*$ is said to be a {\em member of $T$} if there exists a maximal directed path $\lb e_1, e_2, \ldots, e_n \rb$ in $T$ such that $s$ is the concatenation of the letters given on the edges $e_1, e_2, \ldots, e_n$ in the order. 
For such $T$, the algorithm {\sc Query} is naturally generalized. 
We re-define $\hat{s}(m, w)$ in \refeq{Estimate} by
\[
\hat{s} \lp m , w \rp = m + \sum_{j=1}^w 2k^{w-j} (k-1) z_j.
\]
We show the same claim as \reflm{Config1} with the new $\hat{s}(m, w)$. 

\newpage

\begin{lemma} \label{Config2}
Algorithm {\sc Query} on a given query $q$ having a configuration $Z=\lb z_1, z_2, \ldots,  z_w \rb$ takes no more than $\hat{s} \lp m , w \rp$ steps.
\end{lemma}
\begin{proof}
Prove by induction on $w$. The basis $w=1$ is straightforward to check. Assume true for $w-1$ and prove true for $w$. It suffices show that the number of steps required by $q$ is at most $\hat{s} \lp m, w - 1 \rp + (k-1) z_w +  (k-1) \hat{s} \lp z_w, w-1 \rp$ since
\begin{small}
\(
&& 
\hat{s} \lp m, w - 1 \rp + (k-1) z_w +  (k-1) \hat{s} \lp z_w, w-1 \rp
\\ &=&
\lp m + \sum_{j=1}^{w-1} 2 k^{w-j-1}(k-1) z_j \rp + (k-1) z_w +
(k-1) \lp z_w + \sum_{j=1}^{w-1} 2k^{w-j-1}(k-1) z_j \rp
\nexteqline
m+ 2(k-1) z_w + k \sum_{j=1}^{w-1} 2 k^{w-j-1}(k-1) z_j
\nexteqline
m + \sum_{j=1}^w 2k^{w-j}(k-1) z_j
=\hat{s} \lp m, w \rp.
\)
\end{small}

To verify it, wlog let $v$ be the node such that the letters $a_1, a_2, \ldots, a_k$ given on the edges from $v$ correspond to the most significant wildcard in $q$. The algorithm {\sc Query} first chooses $a_1$ as the value of the wildcard, then finds all the members of $T$ matching to $q$. This requires at most $\hat{s} (m, w-1)$ steps by induction hypothesis. Then it backtracks to $v$ in $z_w$ steps to find all the members of $T$ that match to $q$ including $a_2$. It takes at most $z_w + \hat{s} (z_w, w-1)$ steps.

The above repeats $k-1$ times for $a_2, a_3, \ldots, a_k$. Thus the total number of traversal steps required by $q$ is upper-bounded by $\hat{s} \lp m, w - 1 \rp + (k-1) z_w +  (k-1) \hat{s} \lp z_w, w-1 \rp$, completing the proof.
\qed \end{proof}

\medskip

The rest of the proof is the same as for a bitwise trie. We find that the average number of steps required by $q$ is no more than
\(
&&
m + \sum_{\scriptstyle 1 \le z \le m \atop 1 \le j \le w} 2 k^{w-j}(k-1) z p_{z, j}
=
m + \sum_{\scriptstyle 1 \le z \le m \atop 1 \le j \le w} 2j k^{w-j}(k-1)  
\frac{{z \choose j} {m-z \choose w-j}}{{m \choose w}}
\nexteqline
m + \sum_{1 \le j \le w} 2j k^{w-j}(k-1)  
\frac{{m+1 \choose w+1}}{{m \choose w}}
\nexteqline
m+ \frac{2(m+1)}{w+1} \cdot \frac{k^{w+1} - (w+1) k + w}{k-1}.
\)
The bound is tight by the same argument also; for $q$ having a configuration $Z$, {\em Query} takes $\hat{s} \lp m , w \rp$ steps exactly if $T$ is a complete $k$-ary tree. 
This completes the proof of \refth{RadixTree}.

\section{Lookup Response Time with Wildcards in a Distributed Hash Table}

In this section, we show the same asymptotic upper bound $O \lp \frac{2^w m}{w} \rp$ for a DHT $D$. We will verify it through the structural similarity between a bitwise trie and DHT: key search by incremental bit improvement. We first define general terminology on a DHT with related facts in Section 3.1. The second subsection presents a necessary probabilistic assumption general in the aforementioned DHT class ${\cal C}$. In Section 3.3, we show that the probability of correct lookup is $1-  2^{-\Omega \lp m \rp}$ in the DHT Chord with sufficiently many independent keys. 
The $O \lp \frac{2^w m}{w} \rp$ bound will be proved with \refprop{Main} in Section 3.4.

\subsection{Distributed Hash Table and Wildcard Query}

Let $S$ be the {\em key space} for DHT $D$. Suppose it consists of the $m$-bit binary integers so that $|S|=2^m$. A node $v$ in $D$ is labeled by a key denoted by $key(v) \in S$. It is said to be the {\em node key} of $v$, which is typically a large random number such as a hash of the IP address of $v$ or that of a file name. The mapping $v \mapsto key(v)$ is an injection, $i.e.$, there is no other node $v'$ in $D$ such that $key(v')=key(v)$. Information is stored at a node as a pair $<$key, value$>$ called {\em entry}. We denote an entry by $<$$d, r$$>$ where $d \in S$ is its {\em data key}.

The {\em distance from $d \in S$ to $d' \in S$} is written as  $\Delta \lp d, d' \rp$, which is defined by the DHT design. For example, Chord measures $\Delta \lp d, d' \rp$ as $d' - d$ mod $2^m$ evaluated clockwise in the circular ring $0,1, \ldots, 2^{m-1}$  \cite{Chord}. Kademlia measures $\Delta \lp d, d' \rp$ by the XOR metric \cite{Kademlia}.  For a data key $d$, we say that the node $v$ such that $\Delta \lp d, key(v) \rp$ is minimum is the {\em successor of $d$}, and $v$ such that $\Delta \lp key(v), d \rp$ is minimum is the {\em predecessor}. 
An entry $<$$d, r$$>$ is stored at the successor or predecessor of $d$, or in a generalized object to include them. 
Also the {\em successor of $v$} is the node $v' \ne v$ such that $\Delta \lp key(v), key(v') \rp$ is minimum, and {\em predecessor of $v$} is $v' \ne v$ such that $\Delta \lp key(v'), key(v) \rp$ is minimum. 
A {\em neighbor} of $d$ or $v$ is its successor or predecessor. 
Denote by $n$ the number of nodes in $D$. We assume $m = O \lp \log n \rp$ conventionally.

It is called {\em lookup in $D$} to determine the membership of a given key $d \in S$ in $D$, written as $d \in D$ or $d \not\in D$. 
To answer it, the lookup protocol runs at the {\em current peer node} of $D$ moving to another if necessary. A {\em hop} is a change of the current peer node. The average number of hops per lookup is said to be the {\em network diameter of $D$}.

In addition, each node $v$ holds a set of addresses of other nodes determined by certain rules, usually including all the $v$'s neighbors. 
It is called the {\em routing table of $v$}. We may simply say the routing table includes the nodes rather than their addresses. 
A good DHT is designed with a routing table and distance $\Delta \lp d, d' \rp$ that allow for efficient lookups and updates of entries. For $D$ in the aforementioned DHT class ${\cal C}$, the routing table size and network diameter are both  $O \lp \log n \rp$. The class ${\cal C}$ includes the four major DHTs Chord, Pastry, Tapestry and Kademlia.

With the above, a {\em query $q$ to $D$ with uniformly random $w$ wildcards} is defined the same way as to a bitwise trie $T$. We say that a lookup/query is {\em resolved} if the protocol returns the correct answer. 
Let $h$ stand for the average number of hops required to resolve $q$. It is our measure of $q$'s response time. Our goal in this section is to show $h=O \lp \frac{2^w \log n}{w} \rp$ for DHT $D \in {\cal C}$ and $q$ with uniformly random $w$ wildcards.

\subsection{A Probabilistic Assumption for $D$}

It works similarly to a bitwise trie $T$ how to find a data key $d$ that is a member of DHT $D$: by repeatedly moving to a node $v$ such that $key(v)$ has a smaller distance to $d$ than the current peer node. The number of significant bits shared by $d$ and $key(v)$ is increased incrementally. 
In our proof of the $O \lp \frac{2^w m}{w} \rp$ bound for $D$, we need another probabilistic assumption to justify this incremental bit improvement.

A DHT $D$ or its lookup protocol is said to {\em improve at least one bit per hop, correctly with high probability} if it satisfies the three conditions A)--C) below: In finding a target data key $d$, let $v$ be the current peer node and $v_d$ be a neighbor of $d$. 
Suppose that $2^{g-1} \le \Delta \lp key(v), key(v_d) \rp<2^g$ for an integer $g>1$. Let $v'$ be a node in the routing table such that
\beeq{v'Condition}
2^{g-2} \le \Delta \lp key(v'), key(v_d) \rp<2^{g-1}.
\eeq
The three conditions are:
\begin{enumerate} [A)]
\item The routing table of $v$ include $v'$ such that \refeq{v'Condition} with probability $1 - 2^{-m^{1/2+\ep}}$ for some sufficiently small  constant $\ep>0$.
\item If there exists such $v'$, the lookup protocol must move the current peer node to $v'$.
\item  The worst case number of hops for the lookup does not exceed a polynomial in $m$.
\end{enumerate}

If the routing table of $v$ does not include such $v'$, the lookup protocol may decide $d \not \in D$, or change $v$ to another with no guarantee on the closeness to $d$. This {\em error case} occurs with a small probability at most $2^{-m^{1/2+\ep}}$ for each $v$.

Our assumption for the proof of $h=O \lp \frac{2^w \log n}{w} \rp$ is now stated as: 

\medskip

\noindent
{\bf Assumption II: The lookup protocol of the considered DHT improves at least one bit per hop, correctly with high probability.}

\medskip

This property of incremental bit improvement is common in the considered DHT class ${\cal C}$. The lookup protocol keeps improving another bit until $d$ is between the node keys of $v$ and its neighbor for the first time. In the end it identifies both the successor and predecessor of $d$. An error case may occur with probability $p(m) 2^{-m^{1/2+ \ep}} < 2^{-m^{\lp 1+ \ep \rp/2}}$ for some polynomial $p(m)$. Thus any lookup in $D$ satisfying the assumption is resolved with average number of hops $O \lp m \rp = O \lp \log n \rp$, and probability at least $2^{-m^{\lp 1+ \ep \rp/2}}$.

Hence Assumption II is general in ${\cal C}$, and is satisfied by the above four DHTs: One can check that all of their lookup protocols improve at least one bit per hop with high probability. The actual magnitude of the high probability depends on $m$, $n$, and the frequencies of entry updates and routing table maintenance. Assumption II with the bound $1-2^{-m^{1/2+\ep}}$ is true for the four DHTs with some possible performance parameters in practice. Notice that if $m=C \log n$ for a constant $C$, it means $1-2^{-m^{1/2+\ep}} \le  1 - n^{-\ep}$ for any small constant $\ep>0$ and sufficiently large $m$ and $n$. The error probability bound $n^{-\ep}$ can be achieved in any of the four DHTs. Also the condition C) is satisfied by the maximum number of hops allowed for a lookup, which is set in the DHT. 

\medskip

It has been seen that $D$ satisfying Assumption II searches for keys with the same incremental bit improvement as a bitwise trie $T$. Hence we will be able to apply \refprop{Main} to $D$ to show the $O \lp \frac{2^w m}{w} \rp$ bound. Here the following {\em natural query protocol} is assumed for $D$, which is equivalent to the algorithm {\sc Query}.

\medskip

\noindent
{\bf Natural Query Protocol}: {\em First set every $*$ in $q$ as 0 and search for the data key in the DHT. Change the least significant unfinished wildcard from $*=0$ into $*=1$. Search for the new data key started at the current peer node. Repeat until the membership of every desired data key is determined.}

\medskip

Note that we consider two independent probability spaces for a) the key distribution in $D$, and b) the distribution of configurations of $q$. If we say the average number of hops for $q$ in $D$, it means the average over the joint distribution decided by a) and b).

\subsection{The Probability of Correct Lookup in the DHT Chord}

In case $D$ is Chord, we can present a parameter class such that $D$ satisfies Assumption II exactly. Consider the following argument.

\begin{lemma} \label{ChordAndII}
Let $D$ be a distributed hash table Chord defined over the $m$-bit key space $S$ with $n$ nodes where $n$ and $m=O \lp \log n \rp$ are sufficiently large. $D$ satisfies Assumption II if:
\broman
\item there are at least $C m n$ entries stored in $D$ for a sufficiently large constant $C>1$, and
\item an entry $<$$d, r$$>$ is stored at a node chosen with the uniform probability density function, independently of the others\footnote{This statement considers a probability space constructed for each given $m$, $n$ and the number of stored entries. Its event set consists of all the cases of contained node keys and entries. It defines a PDF of 
node choice to store each entry. It is uniform and independent of any other event, as the statement assumes.  
}.
\eroman
\end{lemma}
\begin{proof}
It suffices to show that there are at least $m$ entries stored at any given node $v$ with high probability, which is seen as follows. By the construction of Chord \cite{Chord}, the $i^{th}$ entry stored at $v$ has a pointer to the successor of $key(v)+2^{i-1}$, called {\em finger}. In other words, $v$'s routing table is required to include the address of the successor if there are $i$ entries or more stored at $v$. If there are $m$ entries at $v$ with high probability, its routing table has the finger to the successor of $key(v)+2^{i-1}$ for every $i \le m$. Then the lookup protocol defined by Chord improves at least one bit per hop correctly with high probability\footnote{
If the routing table of $v$ includes no other node $v'$ closer to the desired key $d$ ($i.e.$, such that \refeq{v'Condition}), the protocol of Chord decides $d \not \in D$, rather than performing further lookup with no guarantee to the closeness to $d$. 
}.

Let $N$ be the total number of entries in $D$ that is at least $Cmn$ by Condition i), and $j$ be the number of entries stored at $v$. Due to ii), deciding if the $i^{th}$ entry is stored at $v$ is a Bernoulli trial with probability of success equal to $\frac{1}{n}$. Repeating it $N$ times, we have 
$
Pr(j \le m) = \sum_{j \le m} {N \choose j} \lp \frac{1}{n} \rp^j \lp 1 - \frac{1}{n} \rp^{N-j}, 
$
where $Pr(\cdot)$ denotes the probability of the argument event.
We will show 
\beeq{L5Target}
Pr \lp j \le m \rp < e^{- \frac{C}{2} m}.
\eeq
Then 
$Pr \lp j  \le m \rp < e^{- \frac{C}{2} m} < e^{- m^{1/2+\ep}}$, meaning $v$ has $m$ entries with high probability as required by Assumption II. (Note that the assumption considers a single particular hop from the current peer node $v$.)

We show \refeq{L5Target} by the Chernoff bound given in \cite{Kobayashi}. For our case, it provides the upper bound
\beeqn
&& \label{ChernoffBound}
Pr \big( X_1 + X_2 + \cdots + X_N \le m  \big)
\le \min_{t \le 0} \lb e^{-tm + N \ln M(t)} \rb, 
\\ \textrm{where} && \nonumber
M(t)= \lp 1  - \frac{1}{n} \rp e^{0 \cdot t}+ \frac{1}{n} e^{1 \cdot t}
= 1 + \frac{e^t - 1}{n}.
\eeqn
Here $X_i$ is the random variable that represents the $i^{th}$ Bernoulli trial, $i.e.$, $X_i=1$ if $i^{th}$ entry is stored at $v$ and $X_i=0$ otherwise. Also $M(t)$ is the {\em moment generating function of $X_i$} where $t$ is a real parameter.

By \refeq{ChernoffBound}, 
$
Pr \lp j \le m \rp \le  e^{-tm + N \ln M(t)}
$
for the parameter 
$
t = \ln \frac{mn}{N} \le \ln \frac{1}{C} < 0
$ that is particularly chosen. 
The natural logarithm of the moment generating function is 
$
\ln M(t) = \ln \lp 1 + \frac{\frac{mn}{N} - 1}{n} \rp
$ for this $t$. We now have 
\beeq{Eq1L5}
\ln Pr \lp j \le m \rp < - m \ln \frac{mn}{N} + N \ln \lp 1 + \frac{\frac{mn}{N} - 1}{n} \rp, 
\eeq
desiring that its RHS is at most $-Cm/2$ to show \refeq{L5Target}.

Put $y = \frac{N}{mn} \ge C$ that is sufficiently large. By the Taylor series of the natural logarithm, $\ln \lp 1 + \frac{\frac{mn}{N} - 1}{n} \rp = \ln \lp 1 - \frac{1- \frac{1}{y}}{n} \rp \le  - \frac{1- \frac{1}{y}}{n} + O \lp \frac{1}{n^2} \rp$. By \refeq{Eq1L5} and $N=mny$, 
\(
\ln Pr \lp j \le m \rp &<&
m \ln y  - \frac{N}{n}  \lp 1  - \frac{1}{y} \rp + O \lp \frac{N}{n^2} \rp
\nexteqline
m \lp \ln y  -  y   \lp 1  - \frac{1}{y} \rp + O \lp \frac{y}{n} \rp \rp
\nexteqline
m \lp \ln y  -  y   +  1 + O \lp \frac{y}{n} \rp \rp
<- \frac{y}{2} m \le -\frac{C}{2}m.
\)
This confirms \refeq{L5Target} proving the lemma. 
\qed\end{proof}

Observe that if Chord $D$ satisfies Conditions i) and ii), an error case occurs for each lookup with probability at most $m \cdot e^{- \frac{C}{2} m} < 2^{-\Omega \lp m \rp}$ due to \refeq{L5Target}. Then any lookup in $D$ is resolved correctly with at most $m=O \lp \log n \rp$ hops and probability $1 - 2^{- \Omega \lp m \rp}$. Therefore:

\begin{theorem} \label{CorrectLookup}
Let $D$ be a distributed hash table Chord defined over the $m$-bit key space $S$ with $n$ nodes where $n$ and $m=O \lp \log n \rp$ are sufficiently large. Suppose that it satisfies the following two.
\broman
\item There are at least $C m n$ entries stored in $D$ for a large constant $C>1$. 
\item An entry  $<$$d, r$$>$  is stored at a node chosen with the uniform probability density function, independently of the others.
\eroman
Then any lookup in $D$ is resolved with at most $m=O \lp \log n \rp$ hops and probability $1 - 2^{- \Omega \lp m \rp}$. \qed
\end{theorem}

The theorem confirms the aforementioned bound $1 - 2^{-\Omega \lp m \rp}$. 
In other words, the sufficient condition for a successful lookup in Chord with the probability bound is i) and ii), which assumes that there are enough entries in $D$ created by a series of mutually independent $N$ Bernoulli trials.

\subsection{Proof of the $O \lp \frac{2^w \log n}{w} \rp$ Bound}

We now show our main claim.

\begin{theorem} \label{DHTQueryTime}
Let $D$ be a distributed hash table defined over the $m$-bit key space, which improves at least one bit per hop correctly with high probability, and let $q$ be a query to $D$ with uniformly random $w$ wildcards. Then the natural query protocol resolves $q$ with high probability, and with the average number of hops at most $O \lp \frac{2^w m}{w} \rp$.  
\end{theorem}
\begin{proof}
Denote by $h$ the average number of hops, and by $d_1, d_2, \ldots, d_{2^w}$ the $2^w$ data keys specified by $q$ in the order determined by the natural query protocol. 
We first show
\beeq{eqDHT}
h \le \frac{m+1}{w+1} \lp 2^{w+2} - 2 w - 4 \rp  + m + O \lp 2^{-m^{\lp 1+ \ep \rp/2}} \rp.
\eeq
Observe facts on $h$ and $d_i$. 

\balph
\item At most $m$ hops are necessary to determine if $d_1 \in D$, and $j$ hops to determine if $d_i \in D$ for $i >1$, where $j$ is the position of the unfinished least significant wildcard in $q$ when the lookup for $d_{i-1}$ is complete. By Assumption II, this is true except for an error case occurring with probability $2^{m^{1/2+\ep}}$ or less.
\item In an error case, the total number of hops is bounded by $2^w$ times a polynomial in $m$. Its contribution to $h$ is the $O \lp 2^{-m^{\lp 1+ \ep \rp/2}} \rp$ term in \refeq{eqDHT}. 
We ignore it in the arguments below. 
\item Denote by $h_i$ the number of extra hops required for $d_i$ considered in a). 
To compare it with traversal steps in a bitwise trie $T$, let $s_i$ be the worst case number of extra steps necessary for Algorithm {\sc Query} to determine if $d_i$ is in $T$, after the search for $d_{i-1}$ is complete. We have 
\[
h_i \le s_i \textrm{~for~} i=1,2, \ldots, 2^w: 
\]
If $i=1$ then $h_1=s_1=m$, otherwise $h_i=j$ and $s_i=2j$ where $j$ is the same as in a). 

\item Let $b$ be as given by \refprop{Main}. It upper-bounds the average number of steps required by $q$ in $T$. Thus  
$
\E \lp \sum_{i=1}^{2^w} s_i \rp \le b, 
$
where $\E(\cdot)$ denotes the average of the argument random variable. 
\ealph

\medskip

Hence we have
\[
h = \E \lp \sum_{i=1}^{2^w} h_i \rp \le \E \lp \sum_{i=1}^{2^w} s_i \rp 
\le b 
= \frac{m+1}{w+1} \lp 2^{w+2} - 2 w - 4 \rp  + m, 
\]
proving \refeq{eqDHT}.

It remains show that the natural query protocol resolves $q$ with high probability. If no further bit is improved at the current peer node $v$, the protocol may decide that $d_i \not \in D$ or change $v$ to another with no closeness guarantee. Such an error case occurs with probability at most $2^{- m^{1/2 + \ep}}$ by Assumption II. The total number of hops is at most $2^w$ times a polynomial in $m$, say $p(m)$.  An error case occurs at any peer node with probability no more than 
$
2^{- m^{1/2 + \ep}}  \cdot 2^w p(m)
<
2^{- m^{\lp 1+ \ep \rp/2}}.
$
Therefore, the protocol returns the correct answer to $q$ with probability at least $1-2^{- m^{(1+\ep)/2}}$, a high probability. The theorem follows this statement. 
\qed\end{proof}

As stated in Section 3.1, we assume $m = O \lp \log n \rp$ in a DHT, so the theorem means $h= O \lp \frac{2^w \log n }{w} \rp$ as desired. 
The bound is applicable to Chord, Pastry, Tapestry and Kademlia since they satisfy Assumption II.

We note that the bound could also improve the performance of $2^w$ independent lookups in Koorde \cite{Koorde}: Koorde is a variant of Chord with the use of De Bruijin graph, achieving $O \lp \log n \big/ \log \log n \rp$ hops per lookup with $O \lp \log n \rp$ routing table size. If $2^w$ lookups run independently in Koorde, its number of hops is $O \lp \frac{2^w \log n }{\log \log n} \rp$ whose argument is greater than  $O \lp \frac{2^w \log n }{w} \rp$ when $w$ is sufficiently larger than $\log \log n$.

\section{Concluding Remarks and Open Problems}

We have shown the bound $O \lp \frac{2^w m}{w}  \rp$ for both bitwise tries and distributed hash tables in ${\cal C}$, and $O \lp \frac{k^w m}{w}  \rp$ for a general trie of maximum out-degree $k$. They limit the asymptotic running time required by a partial-match query of length $m$ with $w$ uniformly random wildcards. We also confirmed the probability $1 - 2^{-\Omega \lp m \rp}$ of correct lookup in Chord under the natural assumption.

There are some practical cases to which the obtained results can be applied with the assumption of uniform wildcard occurrences. One such case is data retrieval: Suppose that one searches for data records with $m$ attributes, managed as a trie $T$ such that each attribute takes at most $k$ values. 
An example of such a data record is of form {\tt <college, department, building, title, last name, first name>}. The $m=6$ attributes are hierarchical but with the independent equal probability $\frac{1}{m}$ to be a wildcard in a query. The trie $T$ organizing such data records could be an auxiliary data structure to enhance the search speed. In this situation, wildcards included in a query $q$ occur randomly with the uniform PDF. By \refth{RadixTree}, $q$ takes average $O \lp \frac{k^w m}{w} \rp$ steps rather than $k^w m$.

Further research on this problem could consider query protocols to resolve $q$ with non-uniform probability distributions of wildcard occurrence. It is possible that such a protocol runs at multiple peer nodes simultaneously. It would be interesting to investigate its lookup efficiency.

%In addition, we could be interested in the optimal time bound of the natural query protocol in $D \in {\cal C}$. It seems possible to further improve the $O \lp \frac{2^w \log n}{w} \rp$ bound for some DHTs such as Koorde with the De Bruijn graph, and the variants of Kademlia proposed in \cite{InfocomModel}. On the other hand, if we can construct a parameter class of a DHT in ${\cal C}$ where the bit improvement is exactly one in most cases, then the bound could be shown optimal  for the considered $D$.

\section*{Acknowledgements}

The author is especially thankful to Yan Shvartzshnaider and Max Ott for introducing him to this problem. Also the author would like to thank Professor Hisashi Kobayashi at Princeton University, and Processors Martin F\"urer and Piotr Berman at Penn State for their helpful suggestions.


\begin{thebibliography}{20}

\bibitem{Textbook} Cormen, T.H., Leiserson, C.E., Rivest, R.L., Stein, C.: Introduction to Algorithms, Third Edition. MIT Press (2009).

\bibitem{R76} Rivest, R.L.: Parttial-match retrieval algorithms. SIAM Journal on Computing, vol. 5, pp. 15-50 (1976). 

	
\bibitem{CIP02} Charikar, M., Indyk, P., Panigrahy, R.: New Algorithms for Subset Query, Partial Match, Orthogonal Range Searching, and Related Problems. LNCS, vol.\ 2380, pp.\ 451--463, Springer, Heidelberg (2002)


\bibitem{Sedgewick} Sedgewick, R., Wayne, K.: Algorithms (4th edn.). Addison-Wesley (2011).


\bibitem{Knuth} Knuth, D.E.: The art of computer programming, volume 3. Addison-Wesley (1997).



\bibitem{DMap}  Vu, T., Baid, A., Zhang, Y., Nguyen, T. D., Fukuyama, J., Martin, R. P., Raychaudhuri, D.: DMap: A shared hosting scheme for
dynamic identifier to locator mappings
in the global internet. In: Proc. of the 32nd International Conference on Distributed Computing Systems (ICDCS12). IEEE, pp. 698--707 (2012).


\bibitem {DHTAnalysis} Xu, J., Kumar, A., Yu, X.: On the fundamental tradeoffs between routing table size and network diameter in peer-to-peer networks. IEEE Journal on Selected Areas in Communications, vol. 22, pp. 151--163 (2004).



\bibitem{hyper} Rosenkrantz, W. A.: Introduction to probability and statistics for science, engineering, and finance. Chapman and Hall (2009).



\bibitem{ConcreteMath} Graham, R. L., Knuth, D. E., Patashnik, O.: Concrete mathematics. Addison-Wesley (1994).

\bibitem{Chord}
Stoica, I., Morris, R., Karger, D., Kaashoek, M. F., Balakrishnan, H.: Chord: a scalable peer-to-peer lookup service for internet applications. In: Proc. of the SIGCOMM 2001 Conference (SIGCOMM01). ACM, vol.\ 31, pp.\ 149--160, New York (2001).


\bibitem{Pastry} Rowstron, A., Druschel, P.: Pastry: scalable, decentralized object location and routing for large-scale peer-to-peer systems. In: Proc. of IFIP/ACM International Conference on Distributed Systems Platforms (Middleware). LNCS, vol. 2218, pp.\ 329--350. Springer, Heidelberg (2001).



\bibitem{Tapestry} Zhao, B.Y., Huang, L., Stribling, J., Rhea, S.C., Joseph, A.D., Kubiatowicz, J.D.: Tapestry: a resilient global-scale overlay for service deployment. IEEE Journal on Selected Areas in Communications vol. 22, pp. 41--53 (2004).


\bibitem{Kademlia} Maymounkov, P., Mazi`eres, D.: Kademlia: a peer-to-peer information system based on the XOR metric. In: Proc. of the 1st International Workshop on Peer-to Peer Systems (IPTPS02). LNCS, vol.\ 2429, pp. 53--65. Springer, Heidelberg (2002).



\bibitem{KnuthBook} Graham, R. L., Knuth, D. E., Patashnik, O.:   Concrete mathematics: a foundation for computer science (2nd edn.).
Addison-Wesley (1994). 


\bibitem{Kobayashi} Kobayashi, H., Mark, L.M., Turin, W.: Probability, random processes, and statistical analysis: applications to communications, signal processing, queueing theory and mathematical finance. Cambridge University Press (2012).
 


\bibitem{Koorde}
Kaashoek, M.F., Karger, D.R.: 
Koorde: a simple degree-optimal distributed hash table. In: Proc. 
of the 2nd International Workshop on Peer-to-Peer Systems (IPTPS '03). LNCS, vol.\ 2735, pp.\ 98--107. Springer, Heidelberg (2003).



%\bibitem{InfocomModel} Stutzbach, D., Rejaie, R.: Improving lookup performance over a widely-deployed DHT. In: Proc. of the 25th IEEE International Conference on Computer Communications (INFOCOM06), pp.\ 1-12 (2006)



%CS Bibliography database
%http://liinwww.ira.uka.de/bibliography/index.html

%LNCS reference style
%https://www.citethisforme.com/guides/springer-lecture-notes-in-computer-science-alphabetical/how-to-cite-a-journal
% Example: 
% 1.Lehrman, P., Tully, T.: MIDI for the professional. Amsco Publications, London (1993).

\end{thebibliography}
\end{document}